\documentclass[12pt,centertags,reqno]{amsart}
\usepackage[foot]{amsaddr}
\usepackage{latexsym}
\usepackage[english]{babel}
\usepackage[T1]{fontenc}
\usepackage{float}
\usepackage[utf8]{inputenc}
\usepackage[numbers]{natbib}
\usepackage{amssymb,amsmath}
\usepackage{fancyhdr}
\usepackage{url}
\usepackage{hyperref}
\usepackage{verbatim}
\usepackage{leftidx}
\usepackage{bm}
\usepackage{color,graphicx}
\usepackage{tikz}
\usepackage{bbm}

\usepackage{mathrsfs, mathtools}
\usepackage{stmaryrd}

\usepackage{marvosym}
\mathtoolsset{showonlyrefs}

\usepackage{appendix}

\usepackage{appendix}

\usepackage{soul}

\usepackage{enumerate} 
\usepackage{enumitem} 

\usepackage{booktabs}

\numberwithin{equation}{section} \makeatletter
\renewcommand{\subsection}{\@startsection
{subsection}{2}{0mm}{\baselineskip}{-0.25cm}
{\normalfont\normalsize\bf}} \makeatother

\newtheorem{theorem}{Theorem}[section]
\newtheorem{lemma}[theorem]{Lemma}

\newtheorem{definition}[theorem]{Definition}
\newtheorem{remark}[theorem]{Remark}
\newtheorem{proposition}[theorem]{Proposition}

\newtheorem{assumption}[theorem]{Assumption}

\begin{document}

\author[C.~Ceci]{Claudia  Ceci}
\address{Claudia  Ceci, Department MEMOTEF,
University of Rome Sapienza, Via del Castro Laurenziano, 9,
Rome, Italy.}\email{claudia.ceci@uniroma1.it}
\author[L.~Semerari]{Luca Semerari\,
}
\address{Luca Semerari, Department MEMOTEF,
University of Rome Sapienza, Via del Castro Laurenziano, 9,
Rome, Italy.}\email{ luca.semerari@uniroma1.it }

\title{Explicit Solution to a government debt reduction problem: a stochastic control approach}


\begin{abstract}

We analyze the problem of optimal  reduction of the debt-to-GDP ratio in a stochastic control setting. The debt-to-GDP dynamics are modeled through a stochastic differential equation in which fiscal policy simultaneously affects both debt accumulation and GDP growth. A key feature of the framework is the introduction of a cost functional that captures the disutility of fiscal surpluses and the perceived benefit of fiscal deficits, thus incorporating the macroeconomic trade-off between tighten and expansionary policies. By applying the Hamilton-Jacobi-Bellman approach, we provide explicit solutions in the case of linear GDP response to the fiscal policies. 
We rigorously analyze threshold-type fiscal strategies in the case of linear impact of the fiscal policy and provide closed-form solutions for the associated value function in relevant regimes. A sensitivity analysis is conducted by varying key model parameters, confirming the robustness of our theoretical findings. The application to debt reduction highlights how fiscal costs and benefits influence optimal interventions, offering valuable insights into sustainable public debt management under uncertainty.
\end{abstract}

\maketitle





\section{Introduction}

The management of sovereign debt is one of the most challenging issues in modern macroeconomics, particularly in the context of high uncertainty, recurring shocks, and strict budget constraints. 
The ratio of public debt-to-GDP (gross domestic product) of a country is widely recognized as the key indicator of debt sustainability, and its stabilization is crucial for ensuring fiscal credibility and long-term economic stability (see \cite{empev}, \cite{wheeler} among others). 
The literature on optimal public debt management spans 
traditions. In deterministic models of economic growth, debt sustainability is often studied through intertemporal budget constraints and deterministic policy rules, see \cite{casa}, \cite{fat} and \cite{domar}. However, these approaches neglect uncertainty and sequential decision-making under risk. 
Recent developments have emphasized stochastic control methods, in which optimal fiscal policies are characterized as the solution to infinite-horizon stochastic optimization problems with a given objective functional.
In \cite{cade2} and \cite{cade1}, the authors develop a government debt management model to study the optimal debt ceiling, defined as the maximum level of the debt ratio at which government intervention is not required. That is, if the debt ratio of a country exceeds its debt ceiling, the government should reduce the debt ratio by generating fiscal surpluses; otherwise, the debt is considered to be under control and no intervention is needed. This definition is consistent with the interpretation of the 60\% threshold in the Maastricht Treaty (Article 104c).
An optimal debt reduction problem is analyzed in \cite{sicon2018} within a singular stochastic control framework that accounts for the current inflation rate of the country. The framework is further extended in \cite{callceci}, where the authors consider a partially observable setting in which the GDP growth rate is affected by an unobservable stochastic factor representing the underlying economic conditions. In \cite{cade2, cade1, callceci, sicon2018}, the analysis focuses exclusively on the taxation problem, allowing the government to intervene through austerity policies such as spending cuts and/or tax increases aimed at reducing public debt. By contrast, \cite{sicon2020} introduces a singular stochastic control model in which the control variable consists of both spending cuts and spending increases/public investment. The objective functional includes a cost term designed to capture the beneficial or adverse effects of contractionary and expansionary fiscal policies.
In all the aforementioned contributions, however, the potential impact of fiscal policy on GDP dynamics is not explicitly modeled. Since expansionary fiscal policies may stimulate GDP growth, while contractionary policies may hinder economic output, any attempt to reduce the debt-to-GDP ratio should account for the effect of fiscal policy on both the debt level and the GDP denominator. More recently, \cite{ceci} proposes, for the first time, a stochastic control framework that explicitly incorporates the interaction between fiscal policy and GDP dynamics, while allowing for both tightening and expansionary measures. Although several relevant cases such as debt smoothing and debt reduction are discussed, the model does not include a cost functional associated with the implementation of fiscal measures. In this paper, we build on the framework introduced in \cite{ceci} and study the debt reduction problem by incorporating a cost term in the objective functional, which may generate either adverse or beneficial effects depending on whether the policy is contractionary or expansionary. For the special case of a linear impact of fiscal policy on GDP growth, we derive an explicit solution using the Hamilton–Jacobi–Bellman approach. The presence of fiscal policy costs leads to a threshold-type policy: when the debt-to-GDP ratio is below the threshold, the government adopts the maximum allowed deficit-to-debt ratio, whereas when it exceeds the threshold, the government implements the maximum allowed surplus-to-debt ratio.
The paper is organized as follows. In Section 2, we introduce the model framework, the problem formulation, and prove the key properties of the value function for a general cost function. In Section 3, we specify the form of the cost function, addressing the optimal reduction of the debt-to-GDP ratio problem via the HJB equation approach. In particular, we focus on the special case of a linear impact of fiscal policy on GDP growth, and derive explicit expressions for both the value function and the optimal fiscal strategy.
Finally, in Section 4, we conduct a sensitivity analysis to reinforce our theoretical findings and provide an economic interpretation. Some technical proofs are presented in the Appendix.

\section{Model Framework and Problem Formulation}
Consider a complete probability space $(\Omega, \mathcal{F}, \mathbb{P}, \mathbb{F})$ endowed with a complete and right-continuous filtration $\mathbb{F} := \{\mathcal{F}_t\}_{t \geq 0}$.
 
We start with a model similar to the one proposed in \cite{ceci}. Precisely, the debt-to-GDP ratio $X^{u,x}=\{X_t^{u,x}\}_{t \geq0 }$ is controlled via a fiscal policy $u=\{u_t\}_{t \geq 0}$ 
and satisfied the following stochastic differential equation (SDE)
\begin{equation}
\label{SDE1}
    dX_{t}^{x,u} = X_{t}^{x,u}\left[ \left( r - g(u_t) - u_t \right) dt + \sigma \, dW_t \right], \quad X_{0}^{x,u} = x ,
\end{equation}
where $x > 0$ is the initial debt-to-GDP ratio, $r>0$ denotes the real interest rate on debt, $\sigma>0$ is the volatility, $g(u_t)$ is the GDP growth rate affected by the fiscal policy applied and  $\{W_t\}_{t \geq 0}$ is a standard Brownian motion. 
Unlike the majority of studies in the related literature, such as
\cite{cade1}, \cite{cade2}, \cite{callceci}, \cite{sicon2020}, here the GDP growth rate is influenced by the fiscal policy. 


\begin{remark}
The dynamics in Equation \eqref{SDE1} can be justified as follows. Under a fiscal policy $\{u_t\}_{t\geq 0}$ expressed in term of the debt, the nominal debt $D_t$ evolves at time $t \geq 0$ at rate
$r - u_t$, that is,
\[
dD_t = (r - u_t) D_t \, dt.
\]
At the same time, we assume that the GDP, denoted by the stochastic process $\{\psi_t\}_{t\geq 0}$, follows the stochastic
differential equation
\[
d\psi_t = g(u_t) \psi_t \, dt + \sigma \psi_t \, dB_t,
\]
for some Brownian motion $\{B_t\}_{t\geq 0}$. An application of It\^o's formula to the ratio
$X^{x,u}_t := D_t / \psi_t$ and a change of probability measure yield to the SDE in \eqref{SDE1}.   
 
 Moreover, the dynamic in \eqref{SDE1} can be seen as a stochastic version of the one proposed in classical macroeconomic literature 
 see, e.g. \cite{dornfisher} or \cite{blanch}, where, in addition, uncertainty in the model is introduced. 
However, unlike the classical case where the GDP growth rate is assumed constant, it is here influenced by the fiscal policy.
\end{remark}

We consider the problem faced by a government that aims to determine an optimal fiscal policy for managing the sovereign debt-to-GDP ratio. To this end, we introduce a class of admissible fiscal strategies.

\begin{definition}\label{Admissible}
(\textbf{Admissible fiscal policies}) We denote by $\mathcal{U}$ the family of all the $\mathbb{F}$-predictable and $[-U_1,U_2]$- valued processes $\{u_t\}_{t \geq0}$. Where, $U_1 > 0$ denotes the maximum allowed deficit-to-debt ratio and $U_2 > 0$ is the maximum surplus-to-debt ratio.
\end{definition}
Our goal is to minimize the following objective function, over the class of admissible controls $\mathcal{U}$
\begin{equation}
\label{costfun}
    J(x,u)=\mathbb{E}\bigg[\int_0^{+\infty}e^{-\lambda t}f(X_t^{u,x},u_t) dt \bigg ], \quad x\in(0,+\infty),
\end{equation}
where $\lambda>0$ is the government discounting factor and $f:(0,+\infty)\times[-U_1,U_2] \xrightarrow{}\mathbb{R}$ is the cost function.
From now on we make the following assumption.

\begin{assumption}\label{Ass1}
    The GDP growth rate, $g : [-U_1, U_2] \to \mathbb{R}$, satisfies 
    \[
        \bar{g}_1 \leq g(u) \leq \bar{g}_2 \quad \forall  u \in [-U_1, U_2],
    \]
    for some suitable constants $\bar{g}_1 < 0 < \bar{g}_2$.
\end{assumption}

The SDE in equation \eqref{SDE1} has the following explicit solution

\begin{equation}
\label{solexpl}
X^{u,x}_t = x e^{\int_0^t \left(r - g(u_s) - u_s \right) ds - \frac{1}{2} \sigma^2 t + \sigma W_t}, \quad \forall t \geq 0, \; \mathbb{P}\text{-a.s.}
\end{equation}
In particular, without any fiscal intervention we get that 
\begin{equation}\label{solexpl0}
X^{0,x}_t = x e^{(r - g(0)) t} e^{- \frac{1}{2} \sigma^2 t + \sigma W_t}, \quad \forall t \geq 0, \; \mathbb{P}\text{-a.s.}
\end{equation}
and when $r>g(0)$ the debt-to-GDP explodes, because $$\lim_{t\to + \infty}{\mathbb{E}}[X^{0,x}_t]= \lim_{t\to + \infty} x e^{(r - g(0)) t} = + \infty,$$
this result is consistent with classical results in economics literature and motivates the discussion of the debt-to-GDP reduction problem, which we will address in the sequel.

Under Assumption \ref{Ass1}, for any admissible fiscal policiy $\{u_t\}_{t\geq 0}$ and \( m > 0 \), the following inequalities hold
\begin{equation}
\begin{split}
x^m e^{ -m(\bar{g}_2 + U_2) t - \frac{1}{2} m \sigma^2 t + m \sigma W_t }  \leq \left( X^{u,x}_t \right)^m 
\leq x^m e^{m(R - \bar{g}_1 + U_1) t - \frac{1}{2} m \sigma^2 t + m \sigma W_t }, \quad \forall t \geq 0, \; \mathbb{P}\text{-a.s.}
\end{split}
\label{eq:bounds}
\end{equation}
As a consequence,  taking into account that $ {\mathbb{E}}[e^{m\sigma W_t}]=e^{\frac{1}{2}m^2\sigma^2t}$, we obtain the following upper bound for the expected value of the controlled debt-to-GDP ratio
\begin{equation}\ 
\mathbb{E} \left[ \left( X^{u,x}_t \right)^m \right] \leq x^m e^{\lambda_m t}, \quad \forall t \geq 0,
\label{eq:expectation_bound}
\end{equation}
where 
\begin{equation}
\lambda_m := m(r - \bar{g}_1 + U_1) + \frac{1}{2} m(m - 1) \sigma^2.
\label{eq:lambda_def}
\end{equation}
We now impose some assumptions on the cost function proposed in Eq. \eqref{costfun} and on the discount factor $\lambda$.

\begin{assumption}
    
We assume that
\label{assumption1}
\begin{itemize}
    \item [(i)] there exist constants \( C > 0 \) and \( m \geq 1 \) such that
    \[
    {|f(x, u)| }\leq C (1 + x^m), \quad \forall (x,  u) \in (0, +\infty)  \times [-U_1, U_2];
    \]
    \item [(ii)] \( \lambda > \lambda_m \).
\end{itemize}
\end{assumption} 
Condition $(ii)$ is consistent with the real world, where governments are by far more concerned about the present than
the future. Indeed, the larger the $\lambda$, the more concerned the government is about the present
costs than about the future costs.

As usual in stochastic control framework, the government’s problem can be formalized by introducing the associated value function
\begin{equation}  
\label{valuefun}
v(x) = \inf_{u \in \mathcal{U}} J(x,u), \quad x \in (0,+\infty).
\end{equation}
This function is finite as proved below.
\begin{proposition}
For every admissible strategy \( u \in \mathcal{U} \) and  for any $x \in (0,+\infty)$ we have that $|J(x, u)| < +\infty$.
\end{proposition} 
\begin{proof}
    
Under Assumptions \ref{Ass1} and \ref{assumption1}, and recalling estimate in Eq.\eqref{eq:expectation_bound}, for any \( u \in \mathcal{U} \) and $x \in (0,+\infty)$, we have
\begin{equation}
\begin{split}
    |J(x, u)| \leq 
&\mathbb{E} \left[ \int_0^{\infty} e^{-\lambda t} |f(X_t^{x,u}, u_t)| \, dt \right]\\
&\leq C \, \mathbb{E} \left[ \int_0^{\infty} e^{-\lambda t} \left(1 + (X_t^{x,u})^m \right) dt \right]\\
&\leq C \left[ \int_0^{\infty} e^{-\lambda t} dt + x^m \int_0^{\infty} e^{-(\lambda - \lambda_m)t} dt \right]\\
&=\frac{C x^m}{\lambda - \lambda_m} + \frac{C}{\lambda} < +\infty.
\end{split}
\end{equation}
\end{proof} 
We now  focus on the case where the cost function has {the following form \begin{equation}\label{h} f(x,u)=h(x)+ku
\end{equation}} with $h:(0,+\infty) \rightarrow(0,+\infty)$ such that $h(x)\leq C(1+x^m)$, with $C$ and and 
$k$ positive constants. Clearly, $f(x,u)$ in \eqref{h} satisfies the Assumption \ref{assumption1}.


With this choice, the value function in Eq. \eqref{valuefun} reads as
\begin{equation}\label{P1}
v(x)=\inf_{u\in\mathcal{U}}\mathbb{E} \bigg[\int^\infty_0e^{-\lambda t}h(X_t^{u,x}) dt + k \int^\infty_0 e^{-\lambda t}u_t dt\bigg],
\end{equation}

\begin{remark}\label{etaxi}
The second term inside the expectations \eqref{P1} refers as a cost when $u_t>0$, conversely,  when $u_t<0$ it is considered as a benefit. We can justify  this behavior macro-economically because when the fiscal policy tightens (i.e. is more tax heavy) creates a cost in terms of citizens' discontent, instead when the fiscal policy is more expansive (in the sense that it creates deficit instead of taxation to finance the public expenditure) generates a benefit in term of citizens' content. The same considerations are made in \cite{sicon2020}, where a singular stochastic control model is discussed and the government's actions are described via the pair $\{\eta_t,\xi_t \}_{t\geq 0}$, where $\eta_t$ is the cumulative amount made up to time $t \geq 0 $, of a generic tighten fiscal policy, such as spending cuts or increased taxation, with the objective of reducing the debt-to-GDP ratio and $\xi_t$ is the cumulative amount of a generic expansionary policy such as public investments, that creates deficit, made up to time $t \geq 0$.
\end{remark}

Now we prove some properties of the value function $v(x)$.

\begin{proposition}
\label{p31}
The value function satisfies the following properties
\begin{itemize}
    \item [(i)]  $v(x)\geq-\frac{kU_1}{\lambda}$ and  
    $lim_{x\to 0^+} v(x) \geq -\frac{kU_1}{\lambda}.$
    
    \item [(ii)] $\exists M>0: v(x)\leq M(1+x^m)$.
    \end{itemize} 
If in addition $h(x)$ is increasing:
\begin{itemize}
    \item [(iii)] $v(x)$ is increasing;
    \end{itemize}
    If in addition $h(x)$ is convex:
    \begin{itemize}
    \item [(iv)] $v(x)$ is a convex function.
\end{itemize}
\end{proposition}

\begin{proof}

We introduce $\bar{f}(x,u)=h(x) +ku+kU_1$ and the value function associated to $\bar{f}$
\begin{equation}
\label{aux}
    \bar{v}(x):=\inf_{u \in \mathcal{U}}\mathbb{E}\bigg[\int_0^{+\infty}e^{-\lambda t}(h(X^{u,x}_t)+ k u_t)dt+\int_0^{+\infty}e^{-\lambda t}kU_1dt\bigg].
\end{equation}
We observe that
\begin{equation}
\label{vbarra}
    \bar{v}(x)=v(x)+kU_1(-\frac{e^{-\lambda t}}{\lambda})\bigg|^{+\infty}_0=v(x)+\frac{kU_1}{\lambda}.
\end{equation}
The point $(i)$ is clearly demonstrated because $\bar f\geq 0$ and so $\bar{v}(x)\geq0$ and by the Eq. \eqref{vbarra}  we get  that  $v(x)\geq-\frac{kU_1}{\lambda}$.  Moreover, taking the limit for $x\to 0^+$, point $(i)$ is proved.

 
To prove $(ii)$ we choose in Eq. \eqref{aux}, on the right hand, the null control $u_t \equiv 0$, then by the growth condition $h(x)\leq C(1+x^m)$ and since $f(x,0)=h(x)$ and $\bar{f}(x,0)=h(x)+kU_1$ we have
\begin{equation}
\begin{split}
    \bar{v}(x)&\leq \mathbb{E}\bigg [\int^{+\infty}_0e^{-\lambda t}h(X^{0,x}_t) dt+\frac{kU_1}{\lambda}\bigg]\\
    &\leq \mathbb{E}\bigg [\int^{+\infty}_0e^{-\lambda t}C(1+X^{0,x}_t)^m dt+\frac{kU_1}{\lambda}\bigg].
    \end{split}
\end{equation}
The last term reads as
\[
\mathbb{E}\bigg [\int^{+\infty}_0e^{-\lambda t}C(1+X^{0,x}_t)^m dt+\frac{kU_1}{\lambda}\bigg]= C \mathbb{E}\big [\int^{+\infty}_0e^{-\lambda t}(X^{0,x}_t)^m dt\big]+\frac{kU_1 +C}{\lambda}.
\]
then, from  \eqref{eq:expectation_bound}, we have
\[
\bar{v}(x)\leq\int^{+\infty}_0e^{-\lambda t}Cx^me^{\lambda_mt} dt+\frac{kU_1 +C}{\lambda},
\]
\[
 \bar{v}(x)\leq Cx^m \frac{e^{(\lambda_m-\lambda)t}}{\lambda_ m-\lambda}\bigg |^{+\infty}_0+\frac{kU_1+C}{\lambda},
\]
and so
\begin{equation}
   0 \leq\bar{v}(x)\leq\frac{C}{\lambda-\lambda _m}x^m+ \frac{kU_1+C}{\lambda}.
\end{equation}
Thus, by Eq. \eqref{vbarra} we have that $-\frac{kU_1}{\lambda} \leq v(x) \leq\frac{C}{\lambda-\lambda_m} x^m+\frac{C}{\lambda}$, and so for a suitable $M>0$, $v(x)\leq M(1+x^m)$.

Now, we need to prove the increasing property of $v(x)$, that is $(iii)$.
 Let us take $0 < x \leq x'$. By Eq. \eqref{solexpl}, we observe that for all $u \in \mathcal{U}$,
\[
X_{t}^{x,u} \leq X_t^{u,x'} \quad \forall t \geq 0 \quad \mathbb{P}\text{-a.s.}
\]
and so {using the monotonicity of $h$},  
it follows that
\[
\mathbb{E} \left[ \int_0^{+\infty} e^{-\lambda t} [h(X_t^{u,x})+ku_t] \, dt \right]
\leq 
\mathbb{E} \left[ \int_0^{+\infty} e^{-\lambda t} [h(X_t^{u,x'})+ku_t]\, dt \right]
\quad \forall u \in \mathcal{U};
\]
taking the infimum over $\mathcal{U}$ on both sides completes the proof.

We now prove $(iv)$. Let $x,x'>0$
 and $\theta\in [0,1].$ Defining $x_\theta =\theta x + (1-\theta)x'$, we have that from \eqref{solexpl}, for any $u=\{u_t\}_{t\in [0,T]} \in \mathcal{U}$,  $X^{u,x_\theta}_t =\theta X^{u,x}_t +(1-\theta)X^{u,x'}_t $, $t \in [0,T]$.
 The convexity of $h(x)$ implies that for any $u \in \mathcal{U}$
 \begin{equation}
h(X^{u,x_\theta}_t) \leq 
\theta h(X^{u,x}_t) + (1-\theta) h(X^{u,x'}_t)  
\end{equation}
and as consequence $J(x, u)$ is convex in $x$, that is \begin{equation}\label{ineq:Jconv}
J(\theta x+(1-\theta)x',u)
\le
\theta J(x,u)+(1-\theta)J(x',u).
\end{equation}
Taking the minimum  on  $u\in\mathcal{U}$ in both sides we get
\begin{equation}\label{ineq:inf1}
v(\theta x+(1-\theta)x')
=
\inf_{u\in\mathcal{U}} J(\theta x+(1-\theta)x',u)
\le
\inf_{u\in\mathcal{U}}\big(
\theta J(x,u)+(1-\theta)J(x',u)
\big).
\end{equation}
Finally, since 
\[
\inf_{u\in\mathcal{U}}\big(
\theta J(x,u)+(1-\theta)J(x',u)
\big)
\le
\theta \inf_{u\in\mathcal{U}}J(x,u)
+(1-\theta)\inf_{u\in\mathcal{U}}J(x',u),
\]
we get 
\begin{equation}\label{ineq:inf2}
v(\theta x+(1-\theta)x')
\le
\theta v(x)+(1-\theta)v(x'),
\end{equation}
which concludes the proof.
 
\end{proof}

\section{The Hamilton-Jacobi-Belman equation}
It is well known that when the value function $v(x)$ is sufficiently smooth, is expected to solve the following HJB equation (see for instance, \cite{FS} and \cite{pham} for more details)
\begin{equation}
\label{hjb2}
    \inf_{u\in[-U_1,U_2]}[\mathcal{L}^uv(x)+f(x,u)-\lambda v(x)]=0,
\end{equation}
where $\mathcal{L}^u$ denotes the Markov generator of $X^u$, which, for any constant control $u \in [-U_1, U_2]$, is given by the following differential operator
\begin{equation}
\mathcal{L}^u \phi(x) = x \left[ r - g(u) - u \right]  \phi'(x)
+ \frac{1}{2} \sigma^2 x^2 \phi''(x),
\end{equation}
{with $\phi \in \mathcal{C}^{2}(0, + \infty)$.}

Now we provide a verification theorem based on the classical solution to the HJB \eqref{hjb2}

\begin{theorem}\label{VT}
    Let $w: (0, +\infty) \to \mathbb{R}$ be a function in $\mathcal{C}^2((0, +\infty))$ and suppose
that there exists a constant $C_1 > 0$ such that

\[
|w(x)| \leq C_1 (1 + x^m) \quad \forall x \in (0, +\infty).
\]
If
\[
\inf_{u \in [-U_1, U_2]} \left\{ \mathcal{L}^u w(x) + f(x, u) - \lambda w(x) \right\} \geq 0 \quad \forall x \in (0, +\infty),
\]
then $w(x) \leq v(x)$ for all $x \in (0, +\infty)$.

Now suppose that there exists a measurable function $u^*(x)$ with values in $[-U_1, U_2]$ such that
\[
\inf_{u \in [-U_1, U_2]} \left\{ \mathcal{L}^u w(x) + f(x, u) - \lambda w(x) \right\}
= \mathcal{L}^{u^*} w(x) + f(x, u^*(x)) - \lambda w(x) = 0 \quad \forall x \in (0, +\infty).
\]
Then $w(x) = v(x)$ for all $x \in (0, +\infty)$ and $u^* = \{u^*(X^{u^*, x}_t)\}_{t \geq 0}$ is an optimal (Markovian) control. \end{theorem}

\begin{proof}
 The proof is similar to that of Theorem 4.2 in  \cite{ceci}, with the difference that in \cite{ceci} the cost function $f(x,u)$ is assumed to be nonnegative.
 \end{proof}

We now  focus on the debt reduction problem and so, from now on, we take the cost function as \eqref{h} with $h(x)= C x^m$ with $m\geq 2$ and $C>0$. This disutility function is commonly employed in economic models (see, for instance, \cite{cade2} and \cite{cade1}). The parameter $m \geq 2$ governs the curvature of the function and reflects the degree of aversion to high debt levels: larger values of $m$ imply that the marginal cost of debt rises more sharply as debt increases. The constant $C > 0$ captures the weight the government places on debt minimization relative to other policy objectives.

Using a convex functional form such as $Cx^m$ is particularly appropriate in public debt reduction problems, as it ensures that the cost associated with holding debt increases non-linearly. This setup aligns with the idea that high debt levels may trigger disproportionate economic risks, such as loss of market confidence, rising interest rates, or constraints on fiscal flexibility, making aggressive debt accumulation increasingly costly. The convexity of the disutility function thus creates a strong incentive for the government to contain debt growth, especially as it rises beyond certain thresholds.

Then the HJB in \eqref{hjb2} in this case reads as
 \begin{equation}
   \label{hjb4}  
  \inf_{u\in[-U_1,U_2]}[\mathcal{L}^uv(x)+Cx^m+ku-\lambda v(x)]=0.
 \end{equation}

In the sequel, in order to provide an explicit solution, we study a setting in which fiscal policy exerts a linear influence on GDP growth. Specifically, the GDP growth rate is modeled as
\[
\label{linearimp}
g (u) = g_0 - \alpha u,
\]
where $g_0 \in \mathbb{R}$ and $\alpha>0$ are constants. The term  $g_0 \in \mathbb{R}$ describes the baseline GDP growth rate in the absence of any fiscal intervention. 
Fiscal tightening (i.e. $u_t > 0$) has a contractionary effect on GDP growth, reducing $g(u_t)$, whereas an expansionary fiscal stance (i.e. $u_t < 0$) increases the growth rate.
In this case, the equation \eqref{SDE1} simplifies to
\[
dX_{t}^{x,u} = X_{t}^{x,u} \left[ \left( r - g_0 - (1 - \alpha) u_t \right) dt + \sigma \, dW_t \right], \quad X_{0}^{x,u} = x.
\]
The parameter $\alpha \in (0, +\infty)$ quantifies the relative impact of fiscal policy on GDP growth compared to debt dynamics. When $0 < \alpha < 1$, a fiscal surplus ($u_t > 0$) reduces the instantaneous growth rate of the debt-to-GDP ratio. Conversely, if $\alpha > 1$, even a fiscal deficit ($u_t < 0$) can lead to a decline in the debt-to-GDP growth rate. This behavior arises because, in the first case, fiscal policy has a smaller influence on GDP growth than on debt accumulation, whereas in the second case, it has a stronger impact on GDP than on debt.

Then in this special case, the HJB-equation reduces to 
\begin{equation}\label{hjb3}
\inf_{u\in[-U_1,U_2]}[x \left[ r - g_0 +\alpha u -u \right]  v'(x)
+ \frac{1}{2} \sigma^2 x^2 v''(x)+ C x^m + ku -\lambda v(x)]=0, \quad x \in (0, + \infty).
\end{equation}
Let, $H(x,u,v'(x))= xv'(x)(r-g_0)+[k-(1-\alpha)xv'(x)]u$, Eq. \eqref{hjb3} reads as 
\begin{equation}\label{hjb+h}
    \frac{1}{2}\sigma^2x^2v''(x)+\inf_{u\in [-U_1,U_2]}H(x,u,v'(x))-\lambda v(x)+Cx^m=0.
\end{equation} 
 According to the Verification Theorem \ref{VT}, the candidate to be an optimal control is $u^*_t = u^*(X^{u^*,x}_t)$ where the function $u^*(x)$ is the minimizer of $H(x,u,v'(x))$ and is given by 
 
\begin{equation}
    u^*(x)=\begin{cases} 
\label{u*}
    -U_1 \quad if \quad k-(1-\alpha)xv'(x)\geq0\\
    U_2 \quad if \quad k-(1-\alpha)xv'(x) <0.
    \end{cases}
\end{equation}

We first focus on the case $\alpha>1$. From Eq. \eqref{u*}, the optimal candidate control $u^*(x)$ is constant and given by $u^*(x) \equiv -U_1$ for any $x>0$. Infact, since $v(x)$ is increasing  (see Proposition \ref{p31}) the condition $k-(1-\alpha)xv'(x)\geq0$ can be written as $xv'(x)\geq\frac{k}{1-\alpha}<0$ which is satisfied for any $x>0$.

Hence, for $\alpha>1$, the HJB-equation in \eqref{hjb+h} for any $x>0$ reduces to 
\begin{equation}
\label{hjbbanale}
\frac{1}{2} \sigma^2 x^2 v''(x) 
+ x \left[(1-\alpha)U_1 + r - g_0 ) \right] v'(x)  - \lambda v(x) 
+ C x^m - k U_1 = 0.
\end{equation}

Eq. \eqref{hjbbanale} has the following structure
\begin{equation}
\label{hjbbanale2}
       \frac{1}{2} \sigma^2 x^2 v''(x) + \mu xv'(x)- \lambda v(x) + Cx^m -kU_1= 0,
 \end{equation} 
 with $\mu=r-g_0+(1-\alpha)U_1$, and in the next lemma we derive an explicit solution.
\begin{lemma}\label{lemma1}
The function 
\begin{equation}
    V(x)=C\zeta x^m-\frac{kU_1}{\lambda},
\end{equation}
with
\begin{equation}
\label{zeta2.0}
    \zeta = -\frac{1}{\frac{1}{2} \sigma^2 m^2 + \left( \mu - \frac{1}{2} \sigma^2 \right) m - \lambda},
\end{equation}
solves the Eq. \eqref{hjbbanale2}.
\end{lemma}
\begin{proof}
See Appendix \ref{appendice} for the proof.
\end{proof}

Finally, from Theorem \ref{VT}, we have that  the following result holds. 
\begin{proposition}
For $\alpha>1$ the value function is $v(x)=C\zeta x^m-\frac{kU_1}{\lambda} \quad \forall x>0 $, where $\zeta$ is given in \eqref{zeta2.0} (with $\mu=r-g_0+(1-\alpha)U_1$) and the optimal control is constant and given by $u^*_t \equiv -U_1$, $t\geq0$.
\end{proposition}

We can interpret this result from a macroeconomic perspective because, fiscal policy in the case of $\alpha > 1$ has a greater impact on GDP growth than on debt growth. Therefore, the optimal choice, also in terms of opportunity cost, considering the benefit for citizens (see Remark \ref{etaxi}) and the goal of optimizing the debt-to-GDP ratio can only be a policy of maintaining a constant deficit-to-debt ratio. This is consistent with the result proved in Section 6 in \cite{ceci}  in the case $k=0$. 
 
We now focus on the  case $0<\alpha<1$. By Proposition \ref{p31} we know that $v(x)$ is increasing and convex, and so $xv'(x)$ is increasing. Therefore, as Figure \ref{fig:b} clarifies, from Eq. \eqref{u*} there exists a $b \in (0,+\infty)$ such that
\begin{equation}
 \begin{cases}
          xv'(x)\leq \frac{k}{1-\alpha} <=> 0<x\leq b\\
          xv'(x) > \frac{k}{1-\alpha} <=> x>b.
      \end{cases}   
\end{equation}
\begin{figure}
    \centering
    \includegraphics[width=0.5\linewidth]{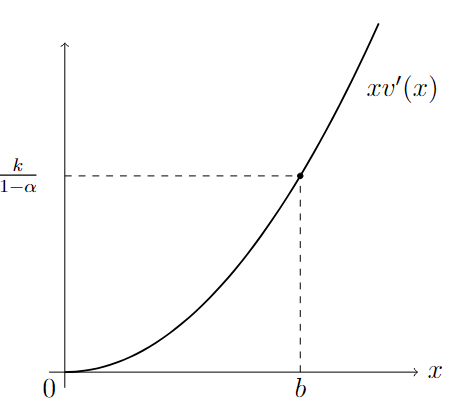}
    \caption{Graphic determination of the threshold b}
    \label{fig:b}
\end{figure}

Thus, for $0<\alpha<1$ the candidate optimal control has the form
\begin{equation}
    u^*(x)=\begin{cases}
    -U_1 \quad for \quad 0<x\leq b\\
    U_2 \quad for \quad x>b,
    \end{cases}
\end{equation}
and the HJB-equation in \eqref{hjb+h}, splits in two ODEs. 

For $0<x\leq b$, we have
  \begin{equation}
  \label{ramo1}
      \frac{1}{2} \sigma^2 x^2 v''(x) + xv'(x) \left[ (1 - \alpha) U_1 + r- g_0 \right]  - \lambda v(x) - k U_1 + C x^m = 0
  \end{equation}
and, for $x > b$, we have
\begin{equation}
\label{ramo2}
    \frac{1}{2} \sigma^2 x^2 v''(x) + x v'(x) \left[ - (1 - \alpha) U_2+r - g_0 \right]  - \lambda v(x) + k U_2 + C x^m = 0.
\end{equation}
We now observe that both equations \eqref{ramo1} and \eqref{ramo2} have the same structure
\begin{equation}
    \label{hjb}
       \frac{1}{2} \sigma^2 x^2 v''(x) + \mu xv'(x)- \lambda v(x) + Cx^m + \bar{U}= 0
 \end{equation} 
with
\begin{equation}
\begin{cases}
\label{mu}
    \mu=r-g_0+(1-\alpha)U_1=\mu_1 \quad for \ Eq.\ \eqref{ramo1} \\
    \mu=r-g_0-(1-\alpha)U_2= \mu_2 \quad for \ Eq.\ \eqref{ramo2}
\end{cases}
\end{equation}

and, $\bar{{U}}=-kU_1$ for Eq.\eqref{ramo1} and $\bar{U}=kU_2$ for Eq. \eqref{ramo2}.
Below we give a preliminary result.

\begin{lemma}
\label{lemma2}
Every function of the form 

\begin{equation}
V(x) = A_1 x^{\gamma_1} + A_2 x^{\gamma_2} + C \zeta x^{m}+\frac{\bar{U}}{\lambda},\\
\end{equation}

with
\[
\gamma_{1} = \frac{-\mu + \frac{1}{2} \sigma^2 - \sqrt{\left( \mu - \frac{1}{2} \sigma^2 \right)^2 + 2 \sigma^2 \lambda}}{\sigma^2},
\]

\[
\gamma_{2} = \frac{-\mu + \frac{1}{2} \sigma^2 + \sqrt{\left( \mu - \frac{1}{2} \sigma^2 \right)^2 + 2 \sigma^2 \lambda}}{\sigma^2},
\]

 \[
 \zeta=-\frac{1}{\frac{1}{2} \sigma^2 (m-1) m + \mu (m-1) - \lambda}
 \]

and $A_1,A_2\in\mathbb{R}$, solves the Eq. \eqref{hjb}.
\end{lemma}
\begin{proof}
See Appendix \ref{appendice} for the proof.
\end{proof}

According to Lemma \ref{lemma2}  the candidate value function $v(x)$ is of the form
\begin{equation}\label{vf1}
v(x)=
    \begin{cases}
        A_{11}x^{\gamma_1}+A_{12}x^{\gamma_2}+C \zeta_1 x^m-\frac{kU_1}{\lambda} \quad 0<x\leq b\\
         A_{21}x^{\bar \gamma_1}+A_{22}x^{\bar \gamma_2}+C \zeta_2 x^m + \frac{kU_2}{\lambda} \quad x>b
    \end{cases}
\end{equation}
where,
\begin{equation}
\begin{split}
\label{gamma}
     \gamma_{1} &= \frac{-\tilde{\mu}_1 - \sqrt{ \tilde{\mu}_1^2 + 2 \lambda \sigma^2 }}{\sigma^2}, \qquad \gamma_{2} = \frac{-\tilde{\mu}_1 + \sqrt{ \tilde{\mu}_1^2 + 2 \lambda \sigma^2 }}{\sigma^2},\\
      \bar{\gamma}_{1} &= \frac{-\tilde{\mu}_2 - \sqrt{ \tilde{\mu}_2^2 + 2 \lambda \sigma^2 }}{\sigma^2},\qquad \bar{\gamma}_{2} = \frac{-\tilde{\mu}_2 + \sqrt{ \tilde{\mu}_2^2 + 2 \lambda \sigma^2 }}{\sigma^2},\\
      \zeta_1&=-\frac{1}{\frac{1}{2} \sigma^2 (m-1) m + \mu_1 (m-1) - \lambda}, \qquad \zeta_2=-\frac{1}{\frac{1}{2} \sigma^2 (m-1) m + \mu_2 (m-1) - \lambda},\\
      \tilde{\mu}_i&=\mu_i-\frac{1}{2}\sigma^2, \quad i=1,2.
   \end{split}
\end{equation}
With $\mu_i$ in Eq. \eqref{mu} and the constants $A_{ij}$, $i,j=1,2$ and the threshold $b$ to be determined. We now make some preliminary considerations.

\begin{lemma} \label{lemma3} We have that $\gamma_1<0$, $\bar{\gamma}_1<0$, and $\bar \gamma_2> m$.  
\end{lemma}
\begin{proof}
See Appendix \ref{appendice} for the proof.
\end{proof}

We are now in the position to give our main result.

\begin{proposition}\label{propvfun}
 Let $0<\alpha<1$ and $m \geq 2$ the value function is 
 
 \begin{equation}\label{vf2}
v(x)=
    \begin{cases}
        A_{1}x^{\gamma_2}+C \zeta_1 x^m-\frac{kU_1}{\lambda} \quad 0<x\leq\ b\\
         A_{2}x^{\bar \gamma_1}+C \zeta_2 x^m +  \frac{kU_2}{\lambda} \quad x>b.
    \end{cases}
    \end{equation}
where $A_1$, $A_2$, and $b$ are given by

\begin{equation}\label{A_1_prop}
A_1 =
\frac{
C(\zeta_1 - \zeta_2) b^{m} (m - \bar{\gamma}_1)
+ \dfrac{\bar{\gamma}_1}{\lambda}\,k(U_1 + U_2)
}{
b^{\gamma_2}(\bar{\gamma}_1 - \gamma_2)
}
\end{equation}

\begin{equation}\label{a2prop}
A_2 =
\frac{
C(\zeta_1 - \zeta_2) b^m (m - \gamma_2)
+ \dfrac{\gamma_2}{\lambda}\,k(U_1 + U_2)
}{
b^{\bar{\gamma}_1}(\bar{\gamma}_1 - \gamma_2)
}
\end{equation}

\begin{equation}\label{b}
b
= \left[
\frac{\bar{\gamma}_1\gamma_2\,k\,(U_1 + U_2)}
{C\,\lambda\,(\bar{\gamma}_1 - m)(\gamma_2 - m)(\zeta_1 - \zeta_2)}
\right]^{1/m}
\end{equation}

 Moreover, the optimal strategy is $u^*_t=u^*(X^{u^*}_t)$, where
 \begin{equation}
    u^*(x)=\begin{cases}
    -U_1 \quad for \quad 0<x\leq b\\
    U_2 \quad for \quad x>b.
    \end{cases}
\end{equation}

\end{proposition}
\begin{proof}
    We start from Eq. \eqref{vf1}. Because of Lemma \ref{lemma2} we choose  $A_{11}=0$, so according to $(i)$ in Proposition \ref{p31}, $lim_{x\to 0^+} v(x) \geq -\frac{kU_1}{\lambda}$
 and $A_{22}=0$ which also implies  condition $(ii)$ in  Proposition \ref{p31}. Therefore, the candidate value function has the following form
\begin{equation}
v(x)=
    \begin{cases}
        A_{1}x^{\gamma_2}+C \zeta_1 x^m-\frac{kU_1}{\lambda} \quad 0<x\leq b\\
         A_{2}x^{\bar \gamma_1}+C \zeta_2 x^m +  \frac{kU_2}{\lambda} \quad x>b.
    \end{cases}
    \end{equation}
    It remains to determine $A_{1}$, $A_{2}$ and $b$ by imposing $v(x) \in \mathcal{C}^2(0, + \infty)$. Precisely, the conditions $v(b^+)=v(b^-)$, $v'(b^+)=v'(b^-)$ and $v''(b^+)=v''(b^-)$ reads as

\begin{equation}
\begin{cases} \label{sitm}
(1) \quad A_2 b^{\bar{\gamma}_1} + C\zeta_2 b^m + \dfrac{k U_2}{\lambda}
= A_1 b^{\gamma_2} + C\zeta_1 b^m - \dfrac{k U_1}{\lambda}, \\
(2) \quad \bar{\gamma}_1 A_2 b^{\bar{\gamma}_1-1} + m C\zeta_2 b^{m-1}
= \gamma_2 A_1 b^{\gamma_2-1} + m C\zeta_1 b^{m-1}, \\
(3) \quad \bar{\gamma}_1(\bar{\gamma}_1-1)A_2 b^{\bar{\gamma}_1-2} + m(m-1)C\zeta_2 b^{m-2}
= \gamma_2(\gamma_2-1)A_1 b^{\gamma_2-2} + m(m-1)C\zeta_1 b^{m-2}.
\end{cases}
\end{equation}

We consider now that $\Delta \zeta := \zeta_1 - \zeta_2.$ Then, the equation $(1)$ in \eqref{sitm} reads as


\begin{equation} \label{eq1mgen}
-A_1 b^{\gamma_2} + A_2 b^{\bar{\gamma}_1}
= C\Delta\zeta\, b^m - \frac{k(U_1 + U_2)}{\lambda}
\end{equation}
and the equation $(2)$ in \eqref{sitm}, reads as 


\begin{equation}\label{eq2mgen}
-\gamma_2 A_1 b^{\gamma_2-1} + \bar{\gamma}_1 A_2 b^{\bar{\gamma}_1-1}
= m C\Delta\zeta\, b^{m-1}.
\end{equation}
 The equations \eqref{eq1mgen} and \eqref{eq2mgen} form linear system in $A_1$ and $A_2$
\begin{equation}\label{sistdet}
\begin{cases}
-b^{\gamma_2} A_1 + b^{\bar{\gamma}_1} A_2
= C\Delta\zeta\, b^m - \dfrac{k(U_1 + U_2)}{\lambda}, \\
-\gamma_2 b^{\gamma_2-1} A_1 + \bar{\gamma}_1 b^{\bar{\gamma}_1-1} A_2
= m C\Delta\zeta\, b^{m-1}.
\end{cases}
\end{equation}

Now we consider the associated coefficient matrix $M$ to the Eq. \eqref{sistdet} and we calculate the determinant $D$ as follow

\begin{equation}
M =
\begin{pmatrix}
- b^{\gamma_2} & b^{\bar{\gamma}_1} \\
- \gamma_2 b^{\gamma_2-1} & \bar{\gamma}_1 b^{\bar{\gamma}_1-1}
\end{pmatrix},
\qquad
D := b^{\bar{\gamma}_1+\gamma_2-1}(\gamma_2 - \bar{\gamma}_1).
\end{equation}
From Lemma \ref{lemma3} we know that $\bar{\gamma}_1 \neq \gamma_2$ leading to $D \ne0$ and so the system has a unique solution for $A_1$ and $A_2$ that can be derived with Cramèr rule:

\begin{equation}\label{a1a2mgen}
\begin{split}
&A_1
= \frac{
C(\zeta_1 - \zeta_2) b^{m} (m - \bar{\gamma}_1)
+ \dfrac{\bar{\gamma}_1}{\lambda}\,k(U_1 + U_2)
}{
b^{\gamma_2}(\bar{\gamma}_1 - \gamma_2)
}\\
&A_2 =
\frac{
C(\zeta_1 - \zeta_2) b^m (m - \gamma_2)
+ \dfrac{\gamma_2}{\lambda}\,k(U_1 + U_2)
}{
b^{\bar{\gamma}_1}(\bar{\gamma}_1 - \gamma_2)
}.\\
\end{split}
\end{equation}



To complete the proof we need to find $b$. For this purpose lets consider $(3)$ from Eq. \eqref{sitm} which reads as



\begin{equation}\label{3condm}
\bar{\gamma}_1(\bar{\gamma}_1-1)A_2 b^{\bar{\gamma}_1-2}
- \gamma_2(\gamma_2-1)A_1 b^{\gamma_2-2}
- m(m-1)C\Delta\zeta\, b^{m-2} = 0.
\end{equation}

Now we replace in Eq. \eqref{3condm} the expression of $A_1$ and $A_2$ derived in \eqref{a1a2mgen}

\begin{equation}\label{deriveb}
- C\,\lambda\,(\bar{\gamma}_1 - m)(\gamma_2 - m)(\zeta_1 - \zeta_2)\, b^m
+ \bar{\gamma}_1\gamma_2\,k\,(U_1 + U_2) = 0.
\end{equation}
Since $\mu_1 > \mu_2$, it follows that $\zeta_1 > \zeta_2$. Moreover, we know that  $(\bar{\gamma}_1 - m)(\gamma_2 - m) < 0$. This condition is satisfied due to Lemma \ref{lemma3}, which states that $\gamma_2 > 2$ and $\bar{\gamma}_1 < 0$. Therefore, from \eqref{deriveb} we end up having
\begin{equation}
b^m
= \frac{\bar{\gamma}_1\gamma_2\,k\,(U_1 + U_2)}
{C\,\lambda\,(\bar{\gamma}_1 - m)(\gamma_2 - m)(\zeta_1 - \zeta_2)},
\end{equation}
which leads to Eq. \eqref{b}.

\end{proof}

\begin{remark}
In the particular case $k=0$ we get from Eq. \eqref{b} that $b=0$ and so the optimal fiscal policy is $u^*_t \equiv U_2$; this  is consistent with the result obtained in \cite{ceci}, where the authors do not consider the cost associated to the fiscal policy. We highlight that the presence of a cost associated to the fiscal policy changes the optimal fiscal policy from a constant maximum taxation level policy to a threshold policy, where the government takes action with the maximum allowed deficit-to-debt ratio if the debt-to-GDP ratio is below the threshold point $b$  and takes action with the maximum allowed surplus-to-debt ratio if the debt-to-GDP ratio is above the threshold point $b$.
\end{remark}

\section{Numerics}

In order to gain quantitative insights into the structure of the optimal debt management problem, we now proceed to compute a numerical approximation of the value function under a baseline parametrization, for the case of interest $0<\alpha<1$. The objective is twofold, at first we want to verify the internal consistency of the closed-form expressions derived in the previous section, and on the other hand we want to provide a concrete benchmark against which alternative parameter scenarios can later be compared. Specifically, we focus on the determination of the switching threshold $b$, the coefficients $(A_1,A_2)$, and the resulting form of the value function $v(x)$. By fixing a plausible set of macroeconomic and fiscal parameters, we can assess both the feasibility of the theoretical solution and its economic interpretability.

We consider the fixed parameters in Table \ref{tab:baseline_params} and plot the value function and the threshold point $b$ in  Fig. \ref{fig:base}.
\begin{table}[h!]
\centering
\begin{tabular}{lc}
\toprule
\textbf{Parameter} & \textbf{Value} \\
\midrule
$r$       & $0.1$ \\
$g_0$     & $0.02$ \\
$\sigma$  & $0.25$ \\
$U_1=U_2$ & $1.00$ \\
$\alpha$  & $0.50$ \\
$\lambda$ & $8.00$ \\
$m$       & $2.00$ \\
$C$       & $1.00$ \\
$k$       & $0.1$ \\
\bottomrule
\end{tabular}
\caption{Baseline scenario parameters.}
\label{tab:baseline_params}
\end{table}

\begin{figure}[h!]
\centering
\includegraphics[width=0.85\textwidth]{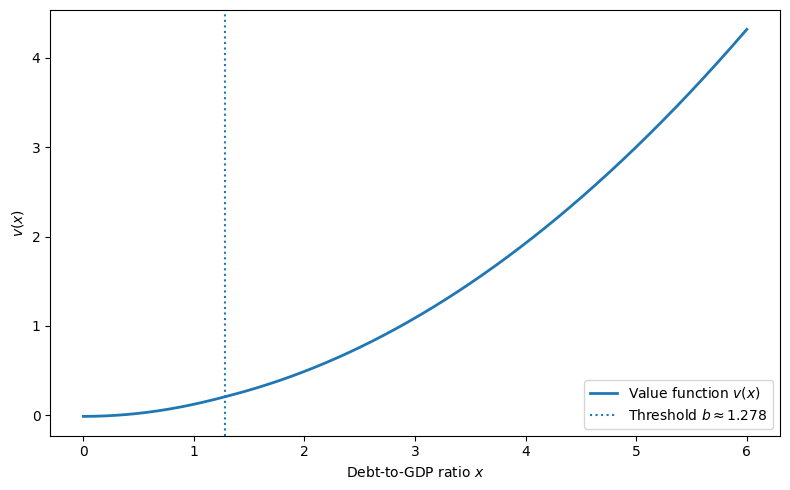}
\caption{Value function $v(x)$ in the baseline case}
\label{fig:base}
\end{figure}

Now, we discuss the economic interpretation  and consistency of results.

\begin{enumerate}
    \item 
    The value function starts from
    \[
    v(0^+) = -\frac{kU_1}{\lambda} = -0.0125 < 0.
    \]
    This small but negative intercept indicates that even in the absence of debt there is a non-zero political or economic cost associated with implementing restrictive fiscal policies. In practice, running a surplus when debt is already null yields no welfare gain and entails a minor political or administrative burden captured by the parameter $k$.

    \item 
    For positive values of $x$, the function $v(x)$ grows smoothly and convexly. 
    This convex shape reflects the fact that the economic cost of debt increases more than proportionally with its level, when debt is low, additional borrowing has limited implications; as debt rises, however, the marginal cost of further increases becomes rapidly larger due to higher risk premia, tighter fiscal constraints, and potential loss of market credibility.

    \item
    The model endogenously identifies a critical debt level $b \approx 1.28$, at which it becomes optimal to switch the fiscal stance:
    \begin{itemize}
        \item if $x<b$, the government optimally maintains an expansionary fiscal policy ($u=-U_1$) to support output and growth, accepting a moderate increase in debt;
        \item if $x>b$, debt becomes excessively costly and the optimal policy shifts to a restrictive stance ($u=U_2$) aimed at debt stabilization.
    \end{itemize}
    Therefore, the threshold $b$ can be interpreted as the economy’s fiscal tolerance level, as long as debt-to-GDP ratio remains below $b$, the marginal benefits of expansionary policies exceed their costs; once $x$ exceeds $b$, the expected marginal cost of debt accumulation dominates.

    \item 
    In the low-debt branch ($x<b$), the homogeneous term (linked to coefficient $A_1$) is almost negligible, and the quadratic particular component $C\zeta_1x^2$ dominates. 
    This means that at moderate debt levels, the value function is primarily shaped by the deterministic cost component associated with debt accumulation.

    \item 
    In the high-debt branch ($x>b$), the coefficient $A_2$ assumes a large absolute value to ensure smooth pasting of the function and its derivatives at $x=b$. 
    This steep curvature in the upper branch reflects the rapid increase in expected costs once the government enters the restrictive regime.
\end{enumerate}

After establishing the baseline evaluation, we now turn to a comparative analysis across heterogeneous economic environments. The aim is to assess how different structural conditions such the interest rate, the baseline growth rate, and the volatility of debt dynamics affect the value function and the optimal switching threshold. 
To this end, as in \cite{cade1} we contrast two polar cases: a strong economy, characterized by lower interest rates, higher growth potential, moderate volatility and a lower cost for the fiscal intervention; and a weak economy, where the opposite conditions prevail. However, unlike \cite{cade1}, our main instrument for highlighting the differences between the two economies is the cost parameter 
$k$, while the parameter $C$ is kept fixed at $1$ throughout the analysis. By recomputing the threshold $b$ and the associated value functions under these two scenarios, we can evaluate both the robustness of the optimal policy rule and the extent to which the shape of $v(x)$ reflects the relative strength or fragility of fundamentals. This comparison provides a clearer interpretation of the model’s economic content, both the threshold level $b$ and the curvature of the value function are lower in the strong economy, revealing how the velocity of the fiscal intervention and the marginal cost of debt accumulation differs substantially between strong and weak macroeconomic contexts.
In order to establish the comparison we fix the following common parameters $U_1 = U_2 = 1,\quad \alpha = 0.5,\quad m=2,\quad C=1, \quad \lambda=8$ and in Table \ref{tab:cmp} the parameters depending on the state of the economy (see Fig. \ref{fig:comp}).
\begin{table}[H]
\centering
\begin{tabular}{lccc}
\hline
\textbf{Parameters} & \textbf{Strong} & \ \textbf{Weak} \\
\hline
$r$  & $0.04$ & $0.29$ \\
$g_0$ & $0.03$ & $0.005$ \\
$\sigma$  & $0.20$ & $0.6$ \\
$k$ & $0.05$& $0.3$\\
\hline
\end{tabular}
\caption{Parameters for Strong/Weak economies.}
\label{tab:cmp}
\end{table}
\begin{figure}[H]
\centering
\includegraphics[width=0.85\textwidth]{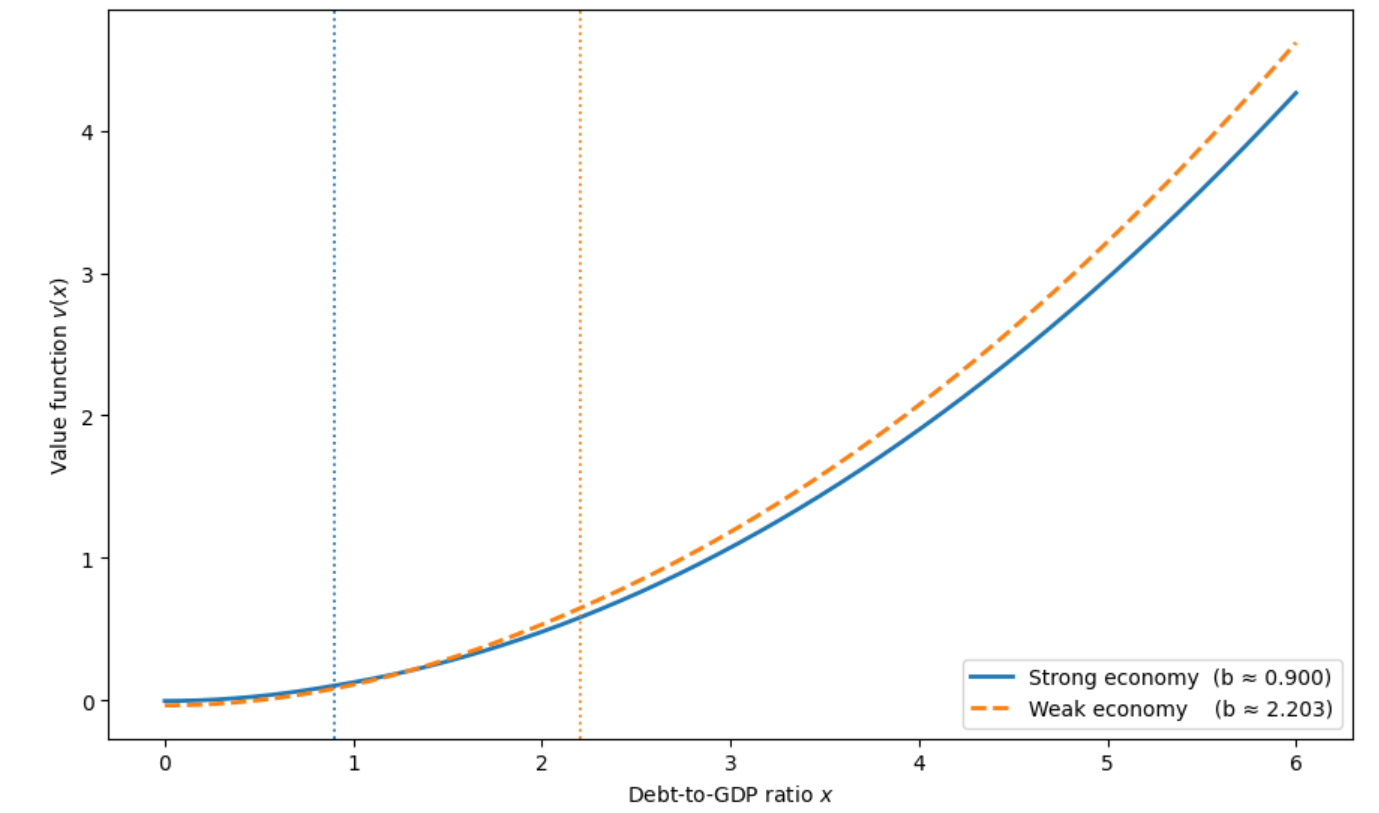}
\caption{Value function $v(x)$ comparison between Strong and Weak economies.}
\label{fig:comp}
\end{figure}

We now provide some economic interpretations of the results.
Figure \ref{fig:comp} compares the value functions and optimal thresholds
for two economies that differ both in their macroeconomic fundamentals 
and in their structural fiscal parameters. 
The resulting patterns are economically consistent and highlight how credibility,
adjustment costs, and policy preferences jointly determine the government’s optimal behaviour.

\begin{enumerate}
    \item 
    The strong economy is characterized by an higher potential growth, low volatility and a relatively low real interest rate 
    $(r,g_0,\sigma)=(0.04,0.03,0.20)$.
    Moreover, the structural parameter $k=0.05$ represent a policy-maker who faces low political adjustment costs, possibly as a consequence of more efficient bureaucratic processes and, more generally, of a more effective and virtuous system of government management.
    Under these favorable conditions, debt dynamics are stable 
    and fiscal corrections are not excessively costly.
    As a result, the optimal debt threshold is low 
    (approximately $b_{\text{strong}}\simeq 0.9$), 
    implying that the government finds it optimal to switch early 
    from expansionary to restrictive fiscal policies.
    The value function is lower across the majority of the debts levels, 
    reflecting the lower debt cost and a more  stable macroeconomic environment.

    \item 
    In contrast, the weak economy displays adverse fundamentals,
    lower growth, high volatility and a high real interest rate 
    $(r,g_0,\sigma)=(0.29,0.005,0.6)$.
    The structural parameter $k=0,3$ represent a policy-maker facing significant political adjustment costs, potentially arising from less streamlined bureaucratic procedures and, more broadly, from institutional arrangements that are not fully optimized in terms of administrative effectiveness.
    Consequently, the optimal threshold is much higher 
    (approximately $b_{\text{weak}}\simeq 2.2$),
    indicating that debt is allowed to accumulate longer 
    before corrective measures are implemented.
    The corresponding value function lies above that of the strong economy,
    consistent with a higher expected cost of debt accumulation.

    \item 
    The shape of the value functions is convex in both cases, 
    confirming that the expected discounted cost of debt increases more than proportionally
    with its level.
    \item 
    Overall, the comparison confirms the theoretical intuition.  
    In strong, credible economies with low fiscal adjustment costs,
    governments act promptly to stabilize debt, achieving long-term sustainability.  
    In weak, volatile economies where  the fiscal adjustment is costly and political horizons are short,
    fiscal reactions are delayed, and the government tolerates higher debt levels 
    before intervention becomes optimal.
\end{enumerate}

\begin{table}[htbp]
\centering
\caption{Parameter configurations for baseline and extreme scenarios.}
\label{tab:parameter_scenarios}
\begin{tabular}{lcccccc}
\toprule
\textbf{Scenario} & $r$ & $g_0$ & $\sigma$ & $k$ & Notes \\ 
\midrule
Strong economy              & $0.04$ & $0.03$ & $0.20$ & $0.05$ & High growth, low volatility \\
Weak economy                & $0.29$ & $0.005$ & $0.60$ & $0.30$ &  Low growth, high volatility \\
High interest               & $0.4$ & $0.02$ & $0.25$ & $0.10$ &  Elevated real rate \\
Low growth (negative $g_0$) & $0.01$ & $-0.8$& $0.25$ & $0.10$ &  Recessionary case \\
Very high volatility        & $0.1$ & $0.02$ & $1.5$ & $0.10$ &  High macro uncertainty \\
Benign macro, high $k$      & $0.04$ & $0.04$ & $0.12$ & $1$ & High adjustment cost \\
Adverse macro, low $k$      & $0.4$ & $0.01$ & $1$ & $0.02$ &  Fragile economy, low cost \\
\bottomrule
\label{tab:b_scenarios}
\end{tabular}
\end{table}

To complete our sensitivity analysis we make some simulation with different scenarios (see Fig.\ref{fig:all_trajectories}). To this purpose, in Table \ref{tab:b_scenarios} we extend the analysis by exploring extreme combinations of macroeconomic and structural parameters. The comparison highlights how the shape of the value function and the location of the optimal threshold $b$ react under limiting conditions. In environments with extremely adverse fundamentals such as high interest rates, low or even negative growth, and very high volatility the threshold $b$ moves upward and the value function becomes steeper. 
This reflects the fact that debt accumulation is more costly and uncertain,
so the government tolerates higher debt levels before committing to restrictive
measures. Conversely, in scenarios with lower a fiscal adjustment cost $k$, the threshold compresses and the value function becomes more curved near the origin, this behavior indicates an earlier fiscal response. Overall, the extreme configurations confirm the internal consistency of the model, as parameters approach their economic limits, the qualitative relationships between macro fundamentals, fiscal preferences,
and optimal intervention timing remain stable.

\begin{figure}[H]
\centering
\includegraphics[width=0.95\textwidth]{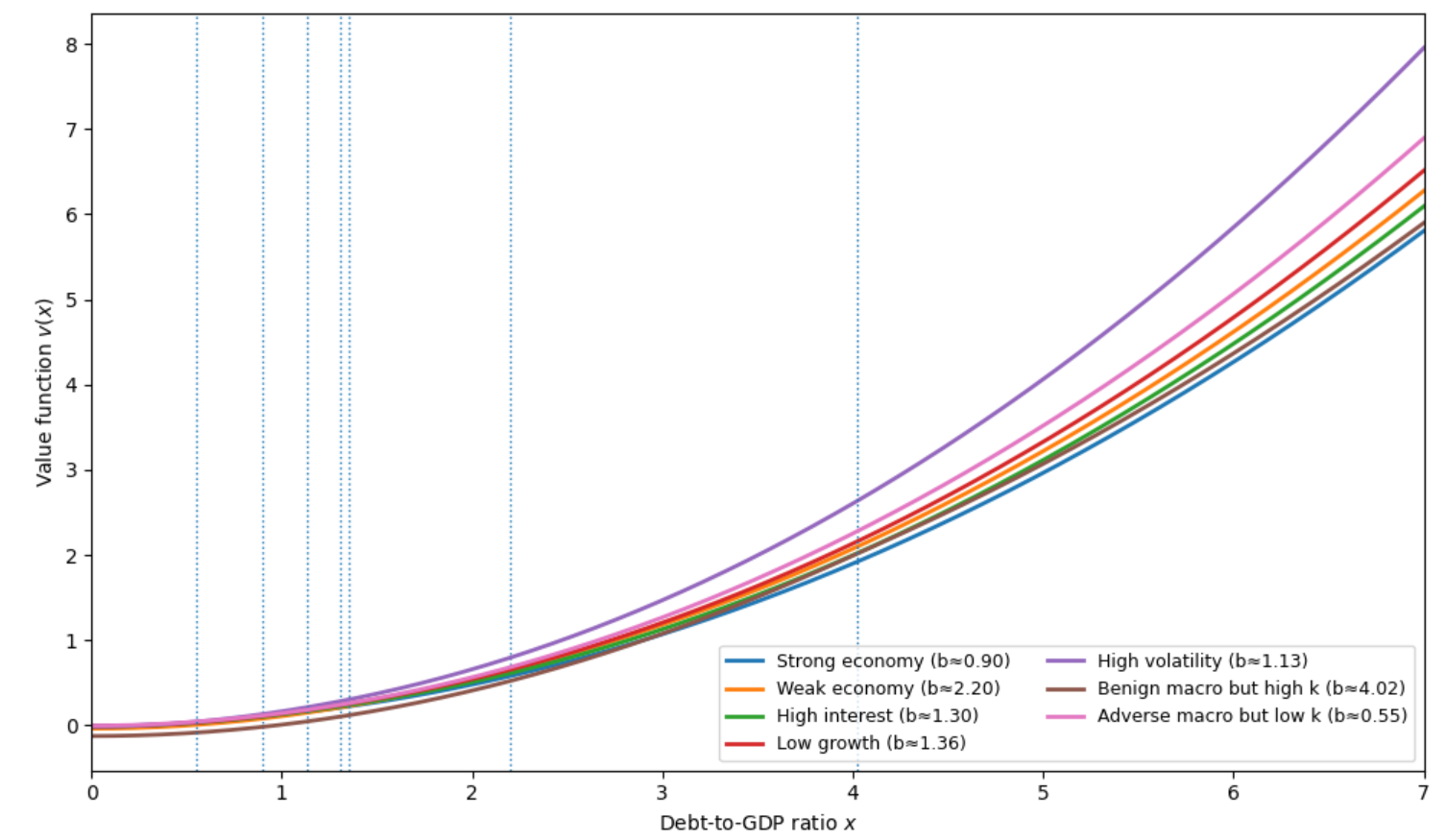}
\caption{Value function $v(x)$ trajectories across scenarios with scenario-specific thresholds \(b\).}
\label{fig:all_trajectories}
\end{figure}

Finally, we conclude the section by displaying in Figure \ref{fig:samplepaths} a representative sample paths of the optimal debt-to-GDP ratio process 
$\{X_t^{u^*}\}_{t\ge0}$, associated with the optimal fiscal policy 
$\{u^*(X_t)\}_{t\ge0}$, and the debt-to-GDP ratio process 
$\{X_t^{0}\}_{t\ge0}$ in the absence of government intervention.
The simulation is performed over a 30-year horizon with baseline parameters in Table \ref{tab:baseline_params}. The threshold level is set at $b = 1,28$, accordingly to our calculations of $b$.
The blue trajectories show how, under the optimal control, the debt ratio initially grows under an expansionary regime ($u_t = -U_1$) and then stabilizes around the threshold as the restrictive policy ($u_t = U_2$) is activated for $X_t > b$. 
Conversely, the orange trajectories represent the uncontrolled process, which grows persistently, highlighting the divergence that occurs when $r > g_0$ in the absence of corrective fiscal action. Overall, the figure illustrates the stabilizing effect of the optimal fiscal policy, which prevents debt divergence and promotes long-term sustainability.

\begin{figure}[H]
    \centering
    \includegraphics[width=0.95\textwidth]{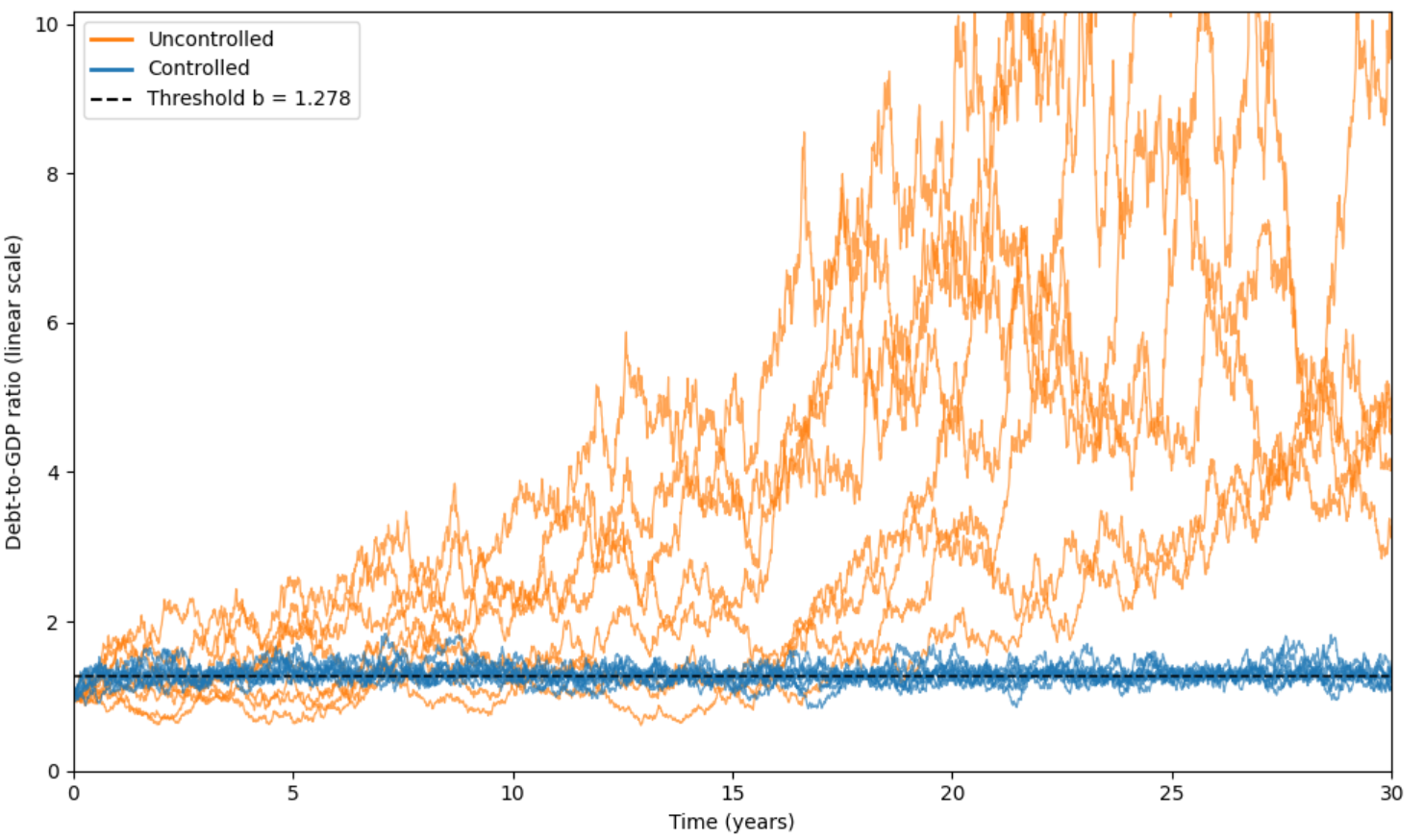}
    \caption{Simulated sample paths of the optimal debt-to-GDP ratio process $\{X_t^{u^*}\}_{t\ge0}$ (blue lines) and the debt-to-GDP ratio process $\{X_t^{0}\}_{t\ge0}$ in the absence of government intervention (orange  lines). 
The dotted line marks the threshold $b \approx 1.28$, separating expansionary ($u_t=-U_1$) and contractionary ($u_t=U_2$) regimes.
}
    \label{fig:samplepaths}
\end{figure}

\section{Conclusions}

This paper examines a stochastic control problem for the optimal management of the public debt-to-GDP ratio in a model that incorporates the dual effect of fiscal policy on both debt accumulation and economic growth. By introducing a cost functional that captures the trade-off between fiscal surpluses and deficits, the framework extends existing approaches to provide a more comprehensive description of sustainable debt dynamics under uncertainty. We  derived explicit and closed-form solutions in the case of a linear fiscal impact on GDP growth, identifying the conditions under which the optimal fiscal policy takes the form of either a constant rule or a threshold-type strategy. When the fiscal multiplier exceeds unity ($\alpha > 1$), the optimal policy is characterized by a constant deficit-to-debt ratio, reflecting the dominance of growth effects over debt accumulation. Conversely, when fiscal policy has a weaker effect on GDP ($0 < \alpha < 1$), the optimal control exhibits a switching behavior; an expansionary stance is optimal for moderate debt levels, while restrictive measures become optimal once the debt-to-GDP ratio surpasses a critical threshold. This threshold, determined in closed form, provides a rigorous quantitative measure of fiscal sustainability.
Numerical evaluations confirm the internal consistency and economic interpretability of the theoretical results. The value function is found to be increasing and convex, reflecting the nonlinear costs associated with high debt levels. Comparative analyses across different economic environments show that while the critical threshold is relatively stable, its curvature and associated welfare losses vary substantially between strong and weak economies, emphasizing the importance of macroeconomic fundamentals. Sensitivity analysis further highlights the pivotal role of the fiscal productivity parameter $\alpha$ and the political cost coefficient $k$ in shaping the optimal consolidation strategy.
Overall, the model offers new insights into the design of optimal fiscal policies under uncertainty. By providing explicit analytical solutions, it enhances the tractability of stochastic control approaches to sovereign debt management and contributes to a deeper understanding of the mechanisms governing debt sustainability, fiscal discipline, and economic resilience.

\begin{appendices}
\section{Proofs of Lemmas 3.2, 3.4 and 3.5}\label{appendice}
\begin{proof}[Proof of Lemma \ref{lemma1}]
We start considering the homogeneous part of the Eq. \eqref{hjbbanale2} (i.e. $U_1=0$) and we we make the following ansatz $\bar V(x)=C\zeta x^m$. Next we calculate the derivatives of $V(x)$
\begin{equation}
    \begin{split}
        \bar V'(x)&=C\zeta m x^{m-1}\\
        \bar V''(x)&=C\zeta(m-1)m x^{m-2}.
    \end{split}
\end{equation}
Substituting into the HJB equation \eqref{hjbbanale2} (with $U_1=0$), we have
\begin{equation}
\label{hjbbanale2.0}
       \frac{1}{2} \sigma^2 x^2 [C\zeta(m-1)m x^{m-2}] + \mu x[C\zeta m x^{m-1}]- \lambda [C\zeta x^m] + Cx^m = 0,
 \end{equation} 

that in terms of \( C x^m \) can be read as
\[
C x^m \left[
\zeta \left( \frac{1}{2} \sigma^2 m(m-1) + \mu m - \lambda \right)
+ 1
\right] = 0,
\]

the characteristic equation associated with the homogeneous terms is

\[
\zeta \left( \frac{1}{2} \sigma^2 m(m-1) + \mu m - \lambda \right) + 1 = 0
\]

from which we find $\zeta$,
\[
\zeta = -\frac{1}{\frac{1}{2} \sigma^2 m^2 + \left( \mu - \frac{1}{2} \sigma^2 \right) m - \lambda}.
\]

Now, we look for a solution $V(x)$ that solves the corresponding non-homogeneous equation
\begin{equation}
\label{hjbbanale3}
\frac{1}{2} \sigma^2 x^2 V''(x) + \mu xV'(x)- \lambda V(x) + Cx^m -kU_1= 0.
\end{equation}
 We look for  $V(x)=\bar{V}(x)+\theta$ which plugging in \eqref{hjbbanale3}
gives
\[
\frac{1}{2} \sigma^2 x^2 \bar{V}''(x) + \mu x\bar{V}'(x)- \lambda \bar{V}(x) + Cx^m - \lambda\theta -kU_1= 0,
\]
and, since $\frac{1}{2} \sigma^2 x^2 \bar{V}''(x) + \mu x\bar{V}'(x)- \lambda \bar{V}(x) + Cx^m=0$, we have that $\theta=\frac{-kU_1}{\lambda}$. Therefore, $ V(x)=\bar{V}(x)-\frac{kU_1}{\lambda}$, and this concludes the proof.\end{proof}

\begin{proof}[Proof of Lemma \ref{lemma2}]
  We start considering the homogeneous part of the Eq. \eqref{hjb} (i.e. $\bar{U}=0$) and in order to find a solution, we make the following ansatz
\begin{equation}
    \begin{split}
&\bar{V}(x) = A_1 x^{\gamma_1} + A_2 x^{\gamma_2} + C \zeta x^{m}\\
&\bar{V}'(x) = A_1 \gamma_1 x^{\gamma_1 - 1} + A_2 \gamma_2 x^{\gamma_2 - 1} + C \zeta mx^{m-1} \\
&\bar{V}''(x) = A_1 \gamma_1 (\gamma_1 - 1) x^{\gamma_1 - 2} + A_2 \gamma_2 (\gamma_2 - 1) x^{\gamma_2 - 2} +C \zeta (m-1)mx^{m-2}
    \end{split}
\end{equation}
Substituting into the HJB equation \eqref{hjb} (with $\bar{U}=0$), we have
\begin{equation}
    \begin{split}
&\frac{1}{2} \sigma^2 x^2 [ A_1 \gamma_1 (\gamma_1 - 1) x^{\gamma_1 - 2} + A_2 \gamma_2 (\gamma_2 - 1) x^{\gamma_2 - 2} +C \zeta (m-1)mx^{m-2}] +\\
&+\mu x [ A_1 \gamma_1 x^{\gamma_1 - 1} + A_2 \gamma_2 x^{\gamma_2 - 1} + C \zeta mx^{m-1}]+ \\
&- \lambda [A_1 x^{\gamma_1} + A_2 x^{\gamma_2} + C \zeta x^{m}] =0
    \end{split}
\end{equation}

Collecting terms by power of \( x \)
\begin{equation}\label{omogenea}
  \begin{split}
&A_1 \left( \frac{1}{2} \sigma^2 \gamma_1 (\gamma_1 - 1) + \mu \gamma_1 - \lambda \right) x^{\gamma_1}\\
&+ A_2 \left( \frac{1}{2} \sigma^2 \gamma_2 (\gamma_2 - 1) + \mu \gamma_2 - \lambda \right) x^{\gamma_2}\\
&+ \left( \bigg (\frac{1}{2} \sigma^2 (m-1) m + \mu (m-1) - \lambda\right)\zeta+1 \bigg)Cx^m  = 0,
\end{split} 
\end{equation}

the characteristic equation associated with the homogeneous terms is

\[
\frac{1}{2} \sigma^2 \gamma (\gamma - 1) + \mu \gamma - \lambda = 0,
\]

and by resolving the second grade equation we get $\gamma_1$ and $\gamma_2$. In a more compact form  we put \( \tilde{\mu} = \mu - \frac{1}{2} \sigma^2 \), so we can write 


\begin{equation}
\begin{split}
\label{gamma1} \gamma_{1} &= \frac{-\tilde{\mu} - \sqrt{ \tilde{\mu}^2 + 2 \lambda \sigma^2 }}{\sigma^2}\\
      \gamma_{2} &= \frac{-\tilde{\mu} + \sqrt{ \tilde{\mu}^2 + 2 \lambda \sigma^2 }}{\sigma^2};
   \end{split}
\end{equation}
with this choice of $\gamma_{1,2}$ the first two terms in Eq. \eqref{omogenea} are equal to zero.

We now look to determine $\zeta$ imposing that
\[
\zeta \bigg( \frac{1}{2} \sigma^2 (m-1) m + \mu (m-1) - \lambda \bigg)+1 =0,
\]
and so, we get the expression for $\zeta$.

Now, we look for a solution $V(x)$ that solves the corresponding non-homogeneous equation
\begin{equation}
\label{cbarra}
\frac{1}{2} \sigma^2 x^2 V''(x) + \mu x V'(x)- \lambda V(x) + Cx^m+\bar{U}= 0.
\end{equation}
 We look for  $V(x)=\bar{V}(x)+\theta$ which plugging in \eqref{cbarra}
gives
\[
\frac{1}{2} \sigma^2 x^2 \bar{V}''(x) + \mu x\bar{V}'(x)- \lambda \bar{V}(x) + Cx^m - \lambda\theta + \bar{U}= 0,
\]
and, since $\frac{1}{2} \sigma^2 x^2 \bar{V}''(x) + \mu x\bar{V}'(x)- \lambda \bar{V}(x) + Cx^m=0$, we have that $\theta=\frac{\bar{U}}{\lambda}$. Therefore, $ V(x)=\bar{V}(x)+\frac{\bar{U}}{\lambda}$, which concludes the proof.
\end{proof}
\begin{proof}[Proof of Lemma \ref{lemma3}]
 $\gamma_1<0$ and $\bar{\gamma_1}<0$ follows immediately from Eq. \eqref{gamma} since is a difference of two negative numbers.
We now focus on $\gamma_2> m$.
So, from Eq. \eqref{gamma1} we have,
\[
\frac{-\tilde{\mu}_2+\sqrt{\tilde{\mu}_2^2+2\sigma^2\lambda}}{\sigma^2}> m
\]
with, $\tilde{\mu}_2=\mu_2-\frac{1}{2}\sigma^2$, that is equal to, 
\[
\lambda> \frac{\sigma^2}{2}m^2+\tilde{\mu}_2m.
\]
However, from Assumption \ref{assumption1} $\lambda>\lambda_m$ and so it is sufficient to prove that $\lambda_m\geq\frac{\sigma^2}{2}m^2+\tilde{\mu_2}m$.

We recall, that $ \lambda_m=m(r-\bar{g_1}+U_1)+m(m-1)\frac{\sigma^2}{2}$, $\bar{g_1}=g_0-\alpha U_2$
so we have
\[
m(r-g_0+\alpha U_2+U_1)+m(m-1)\frac{\sigma^2}{2}\geq\frac{\sigma^2}{2}m^2+\bigg(\mu_2-\frac{1}{2}\sigma^2\bigg)m.
\]
By Eq. \eqref{mu} and some simple arithmetic steps we end up having
\begin{equation}
    U_1\geq-U_2
\end{equation}
which is always verified, and that completes the proof.
\end{proof}
\end{appendices}

\newpage
\bibliography{reference}

@article{ceci,
  author = {Brachetta, M. and Ceci, C.},
  title = {A Stochastic Control Approach to Public Debt Management},
  journal = {Mathematics and Financial Economics},
  year = {2022}
}

@article{sicon2020,
  author = {Ferrari, G. and Rodosthenous, N.},
  title = {OPTIMAL CONTROL OF DEBT-TO-GDP RATIO IN AN N -STATE REGIME SWITCHING ECONOMY},
  journal = {SIAM J. Control Optim.},
  year = {2020}
}

@book{pham,
  author = {Huyen, P.},
  title = {Continuous Control and Optimization with Financial Applications},
  publisher = {Springer},
  year = {2009}
}

@article{cade1,
  author = {Cadenillas, A. and Huamàn-Aguilar, R.},
  title = {On the failure to reach the optimal government debt ceiling},
  journal = {Risks},
  year = {2018}
}

@article{empev,
  author = {Stoilova, D. and Patonov, N.},
  title = {AN EMPIRICAL EVIDENCE FOR THE IMPACT OF TAXATION ON ECONOMY GROWTH IN THE EUROPEAN UNION},
  journal = {Tourism \& Management Studies},
  year = {2013}
}

@article{wheeler,
  author = {Wheeler, G.},
  title = {Sound Practice in Governement Debt Management},
  journal = {The World Bank},
  year = {2004}
}

@article{callceci,
  author = {Callegaro, G. and Ceci, C. and Ferrari, G.},
  title = {Optimal Reduction of Public Debt under Partial Observation of the Economic Growth},
  journal = {Finance and Stochastics},
  year = {2020}
}

@book{blanch,
  author = {Blanchard, O. J. and Fisher, S.},
  title = {Lectures on Macroeconomics},
  publisher = {MIT Press},
  year = {1989}
}

@book{dornfisher,
  author = {Dornbusch, F. and Fischer, S. and Richard, S. and Canullo, G. and Pettenati, P.},
  title = {Macroeconomia},
  publisher = {Mc Graw Hill},
  year = {2014}
}

@article{cade2,
  author = {Cadenillas, A. and Huamàn-Aguilar, R.},
  title = {Explicit formula for the optimal government debt ceiling},
  journal = {Annals of Operations Research},
  year = {2015}
}

@article{casa,
  author = {Casalin,  F. and Dia, E. and Hallett, A.H.},
  title = {Public debt dynamics with tax revenue constraints},
  journal = {Economic Modelling},
  year = {2019}
}

@article{domar,
  author = {Domar, E. D.},
  title = {The "burden of the debt" and the national income},
  journal = {The American Economic Review, 34(4):798–827},
  year = {1944}
}

@article{fat,
  author = {Fatás, A. and Summers, L. H. },
  title = {The permanent effects of fiscal consolidations},
  journal = {Journal of International Economics, 112:238 – 250},
  year = {2018}
}

@book{FS,
  author = {Fleming, W. H. and Soner, H. M.},
  title = {Controlled Markov Processes and Viscosity Solutions},
  publisher = {Springer},
  year = {2006}
}

@article{sicon2018,
  author = {Ferrari, G.},
  title = {ON THE OPTIMAL MANAGEMENT OF PUBLIC DEBT:
A SINGULAR STOCHASTIC CONTROL PROBLEM},
  journal = {SIAM J. Control Optim.},
  year = {2018}
}
\bibliographystyle{plain}
\end{document}